\documentclass[11pt]{article}
\usepackage[top=1in, bottom=1in, left=1in, right=1in]{geometry}
\usepackage{amssymb, amsmath, amsthm, xspace, url, hyperref, enumitem, txfonts}
\usepackage{cleveref}

\newtheoremstyle{newthm}
  {3pt} 
  {3pt} 
  {\itshape} 
  {} 
  {\bfseries} 
  {.} 
  {.5em} 
  {} 
\theoremstyle{newthm}   

\newtheorem{theorem}{Theorem}
\newtheorem*{theorem*}{Theorem}

\newtheorem{lemma}[theorem]{Lemma}
\newtheorem{corollary}[theorem]{Corollary}

\newtheorem{fact}[theorem]{Fact}{\bfseries}{\itshape}

\newtheorem{problem}{Problem}
\newtheorem*{problem*}{Problem}

\theoremstyle{definition}
\newtheorem{definition}[theorem]{Definition}

\theoremstyle{remark}

\newtheorem{example}{Example}

\DeclareMathOperator{\pre}{pre}
\DeclareMathOperator{\cost}{cost}
\DeclareMathOperator{\rk}{rank}
\DeclareMathOperator{\dom}{dom}
\DeclareMathOperator{\IC}{IC}
\DeclareMathOperator{\polylog}{polylog}
\DeclareMathOperator{\PRIV}{PRIV}

\renewcommand{\epsilon}{\varepsilon}
\newcommand{\pub}{\textrm{pub}}
\newcommand{\R}{\operatorname{R}}
\newcommand{\Rs}{\R^s}
\newcommand{\Rpubd}{\R_{\delta}}
\newcommand{\Rspubd}{\R^s_{\delta}}
\newcommand{\D}{\operatorname{D}}

\newcommand{\cond}{\;\ifnum\currentgrouptype=16 \middle\fi|\;}
\newcommand{\rankp}{{\rk}}
\newcommand{\size}{{\operatorname{size}}}

\newcommand{\tr}{{\operatorname{tr}}}
\newcommand{\diag}{{\operatorname{diag}}}

\newcommand{\ang}[1]{{\langle{#1}\rangle}}

\newcommand{\etal}{{\it et al.}\xspace}

\newcommand{\trans}{^\mathrm{T}}

\newcommand{\scrank}{\textsc{Rank}}
\newcommand{\scinv}{\textsc{Inverse}}
\newcommand{\scip}{\textsc{IP}}
\newcommand{\scham}{\textsc{Ham}}

\newcommand{\scsls}{\textsc{LinSolve}}

\newcommand{\F}{\ensuremath\mathbb{F}}
\newcommand{\cP}{\ensuremath\mathcal{P}}
\newcommand{\bx}{\mathbf{x}}

\title{%
  On The Communication Complexity of Linear Algebraic Problems
  in the Message Passing Model
}

\author{
Yi Li\thanks{%
  Max-Planck Institute for Informatics.
  \texttt{yli@mpi-inf.mpg.de}.}
\and
Xiaoming Sun\thanks{%
  Institute of Computing Technology, Chinese Academy of Sciences.
  \texttt{sunxiaoming@ict.ac.cn}.}
\and
Chengu Wang\thanks{%
  Google Inc.
  \texttt{wangchengu@gmail.com}.}
\and
David P. Woodruff\thanks{%
  IBM Almaden Research Center.
  \texttt{dpwoodru@us.ibm.com}.}
}

\setlist{topsep=2pt, itemsep=0pt, parsep=0pt, partopsep=0pt}

\begin{document}

\date{}
\maketitle
\thispagestyle{empty}

\begin{abstract}

We study the communication complexity of linear algebraic
problems over finite fields in the multi-player message passing model, proving a number of
tight lower bounds. Specifically, for a matrix which is distributed among a number
of players, we consider the problem of determining
its rank, of computing entries in its inverse, and 
of solving linear equations. We also consider related problems such as 
computing the generalized inner product of vectors held on different
servers. We give a general framework for reducing these
multi-player problems to their two-player counterparts, showing that the randomized
$s$-player communication complexity of these problems is at least $s$ times the 
randomized two-player communication complexity. Provided the problem has a
certain amount of algebraic symmetry, which we formally define, we can 
show the hardest input distribution is a symmetric distribution, 
and therefore apply a recent multi-player lower bound technique
of Phillips \etal Further, we give new 
two-player lower bounds for a number of these problems. In particular, our
optimal lower bound for the two-player version of the matrix rank problem
resolves an open question of Sun and Wang.

A common feature of our lower bounds is that they apply even to the 
special ``threshold promise'' versions of these problems, wherein the
underlying quantity, e.g., rank, is
promised to be one of just two values, one on each side of some critical
threshold. These kinds of promise problems are commonplace in the
literature on data streaming as sources of hardness for reductions
giving space lower bounds.


\end{abstract}

\newpage

\addtocounter{page}{-1}

\section{Introduction} \label{sec:intro}

{\bf Communication complexity}, introduced in the celebrated work of
Yao~\cite{Yao79}, is a powerful abstraction that captures the essence of
a host of problems in areas as disparate as data structures, decision
trees, data streams, VLSI design, and circuit
complexity~\cite{kushilevitz1997communication}. It is concerned with
problems (or {\em games}) where an input is distributed among $s \ge 2$
players who must jointly compute a function $f:X_1\times\cdots\times
X_s\to Z$, each $X_i$ and $Z$ being a finite set: Player $i$ receives an
input $x_i\in X_i$, the players then communicate by {\em passing
messages} to one another using a predetermined protocol $\cP$, and
finally they converge on a shared output $\cP(x_1,\ldots,x_s)$. 
The main goal of the players is to minimize the
amount of communication, i.e., the total length of messages
communicated.  Put $\bx = (x_1,\ldots,x_s)$. We say that a deterministic
protocol $\cP$ computes $f$ if $\cP(\bx)=f(\bx)$ for all inputs $\bx$.
In a randomized protocol, the players can flip coins and send messages
dependent on the outcomes; we shall focus on the public coin variant,
wherein the coin flip outcomes are known to all players.%
\footnote{Though the private coin model may appear more ``natural,'' our
key results, being lower bounds, are {\em stronger} for holding in the
more general public coin model. In any case, for the particular problems
we consider here, the private and public coin models are asymptotically
equivalent by a theorem of Newman~\cite{Newman91}.} We say a randomized
protocol $\cP$ computes $f$ with error $\delta$ if $\Pr[\cP(\bx) =
f(\bx)] \ge 1-\delta$ for all inputs $\bx$. In all cases, we define the
{\em cost} of $\cP$ to be the maximum number of bits communicated by
$\cP$ over all inputs. We define the {\em deterministic}
(resp.~$\delta$-error {\em randomized}) communication complexity of $f$,
denoted $\D(f)$ (resp.~$\Rpubd(f)$) to be the minimum cost of a protocol
that computes $f$ (with error $\delta$ in the randomized case).
It holds that $\D(f)\ge \Rpubd(f)$ for all $f$ and $0\le\delta\le 1$.

Most work in communication complexity has focused on the two-player
model (the players are named Alice and Bob in this case), which already
admits a deep theory with many applications. However, one especially
important class of applications is {\em data stream}
computation~\cite{raghavan1999computing,munro1980selection}: the input
is a very long sequence that must be read in a few {\em streaming
passes}, and the goal is to compute some function of the input while
minimizing the memory (storage space) used by the algorithm. Several
data stream lower bounds specifically call for {\em multi-player}
communication lower bounds~\cite{DBLP:journals/jcss/AlonMS99}. Moreover,
several newer works have considered distributed computing problems with
streamed inputs, such as the {\em distributed functional monitoring}
problems of Cormode \etal~\cite{CormodeMY08}: in a typical scenario, a
number of ``sensors'' must collectively monitor some state of their
environment by efficiently communicating with a central ``coordinator.''
Studying the complexity of problems in such models naturally leads one
to questions about multi-player communication protocols.

In the multi-player setting, strong lower bounds in the message passing
model\footnote{In contrast to the message passing model is the {\em
blackboard} model, where players write messages on a shared blackboard.}
are a fairly recent achievement, even for basic problems. For the {\sc
SetDisjointness} problem, a cornerstone of communication complexity
theory, two-player lower bounds were long
known~\cite{journals/siamdm/KalyanasundaramS92,razborov1992distributional}
but an optimal multi-player lower bound was only obtained in the very
recent work of Braverman \etal~\cite{BEOPV13}. For computing bit-wise AND,
OR, and XOR functions of vectors held by different parties, as well as
other problems such as testing connectivity and computing coresets for
approximating the geometric width of a pointset, optimal lower bounds 
were given in \cite{PVZ12}. For computing a number of graph
properties or exact statistics of databases, a recent work achieved optimal
lower bounds \cite{wz13}. There are also recent
tight lower bounds for approximating frequency moments \cite{wz12} and 
approximating distinct elements \cite{wz14}. Our chief motivation is
to further develop this growing theory, giving optimal lower bounds for
other fundamental problems.

{\bf Linear algebra} is a fundamental area in pure and applied
mathematics, appearing ubiquitously in computational applications.  The
communication complexity of linear algebraic problems is therefore
intrinsically interesting. The connection with data streaming adds
further motivation, since linear algebraic problems are a major focus of
data stream computation.  Frieze, Kannan and
Vempala~\cite{frieze1998fast} developed a fast algorithm for the
low-rank approximation problem. Clarkson and
Woodruff~\cite{DBLP:conf/stoc/ClarksonW09} gave near-optimal space
bounds in the streaming model for many linear algebra problems, e.g.,
matrix multiplication, linear regression and low rank approximation.
Muthukrishnan~\cite{muthukrishnan2005data} asked several linear algebra
questions in the streaming model including rank-$k$ approximation,
matrix multiplication, matrix inverse, determinant, and eigenvalues.
S\'arlos~\cite{sarlos2006improved} gave upper bounds for many
approximation problems, including matrix multiplication, singular value
decomposition and linear regression. 

\paragraph{Our Results:}  Let us first describe the new two-player communication complexity
results proved in this work. We then describe how to extend these to obtain our multi-player
results. 
\\\\
{\it Two-Player Lower Bounds:}
We start by studying the following closely related matrix
problems. In each case, the input describes a matrix $z \in M_n(\F_p)$,
the set of $n\times n$ matrices with entries in the finite field $\F_p$
for some prime $p$. 
\begin{itemize}
  \item Problem $\scrank_{n,k}$: Under the promise that $\rankp(z) \in
  \{k,k+1\}$, compute $\rankp(z)$.
  \item Problem $\scinv_n$: Under the promise that $z$ is invertible,
  decide whether the $(1,1)$ entry of $z^{-1}$ is zero.
  \item Problem $\scsls_{n,b}$: Under the promise that $z$ is
  invertible, for a fixed non-zero vector $b\in \mathbb{F}_p^n$,
  consider the linear system $zt=b$ in the unknowns $t\in \mathbb{F}_p^n$. 
  Decide whether $t_1$ is zero.
\end{itemize}
There are two natural ways to split $z$ between Alice and Bob. In the
concatenation model, Alice and Bob hold the top $n/2$ rows and bottom
$n/2$ rows of $z$, respectively. In the additive split model, Alice and
Bob hold $x,y\in M_n(\F_p)$ respectively, and $z = x + y$. The two
models are equivalent up to a constant factor~\cite{sun2012randomized},
see \Cref{sec:la}. All of this generalizes in the obvious manner to the
multi-player setting.

\begin{theorem} \label{thm:rank-lb}
  Let $f$ be one of $\scrank_{n,n-1}$, $\scinv_n$, or $\scsls_{n,b}$.
  Then $\R_{1/10}(f) = \Omega(n^2\log p)$.
\end{theorem}
The above immediately implies $\Omega(n^2 \log p)$ space lower bounds
for randomized streaming algorithms for each of these problems, where
the input matrix $z$ is presented in row-major order. See Appendix
\ref{sec:stream} for details. Clearly these
lower bounds are optimal, since the problems have trivial $O(n^2\log p)$
upper bounds, that being the size of the input.
We remark that \Cref{thm:rank-lb} in fact extends to the {\em quantum}
communication model, a generalization of randomized communication that
we shall not elaborate on in this paper. 

To prove these lower bounds, we use the \emph{Fourier witness method}~\cite{sun2012randomized} 
for the promised rank problem, then reduce it to other problems. The reduction to the other problems
critically uses the promise in the rank problem, for which establishing a lower bound was posed as
an open question in \cite{sun2012randomized}. 
Roughly speaking, the Fourier witness method is a special type of dual norm method~\cite{linial2009lower,conf/stoc/Sherstov08,shi2007quantum}.
In the dual norm method, there is a \emph{witness} (a feasible solution of the dual maximization problem for the approximate norms). 
A typical choice of witness is the function itself (such as in the discrepancy method). 
In the \emph{Fourier witness method} the witness is chosen as the Fourier transform of the function. 
This method works well for plus composed functions. For details, see Section~\ref{sec:fourierwitness}.

We also consider the inner product and Hamming weight problems.
Alice and Bob now hold {\em vectors} $x$ and $y$.
\begin{itemize}
  \item Problem $\scip_{n}$: Under the promise that
  $\ang{x,y}\in\{0,1\}$, compute $\ang{x,y}$. Here $x,y\in\F_p^n$.
   \item Problem $\scham_{n,k}$: Under the promise that
   $\|x+y\|\in\{k,k+2\}$, compute $\|x+y\|$. Here $x,y\in\F_2^n$ and
   $\|z\|$ denotes the Hamming weight of $z$, i.e., the number of $1$
   entries in $z$. Note that $x-y = x+y$.
\end{itemize}
We do not provide new two-player lower bounds for $\scip_{n}$ and $\scham_{n,k}$, but state the known ones 
here for use in our $s$-player lower bounds. 
It is known that $R_{1/3}(\scip_{n}) = \Omega(n \log p)$ \cite{SWY12}, and
$\R_{1/3}(\scham_{n,k}) = \Omega(k)$ \cite{hszz06}. 
%
%
%
\\\\
{\it $s$-Player Lower Bounds:}
For each of the above problems, there are natural $s$-player variants. We
consider the {\it coordinator model} in which there is an additional
player, called the coordinator, who has no input. We require that
the $s$ players can only talk to the coordinator. The message-passing
model can be simulated in the coordinator model since every time a Player 
$i$ wants to talk to a Player $j$, Player $i$ can first send a message to the
coordinator, and then the coordinator can forward the message to 
Player $j$. This only affects the communication by a factor of $2$. See, e.g.,
Section 3 of \cite{BEOPV13} for a more detailed description. 

For the matrix problems, Player $i$ holds a matrix $x^{(i)}$ and the
computations need to be performed on $z = x^{(1)} + \cdots + x^{(s)}$.
The Hamming weight problem is similar, except that each $x^{(i)}$ is a
vector in $\F_2^n$. For the inner product problem, each
$x^{(i)}\in\F_p^n$ and we consider the generalized inner product,
defined as $\sum_{j=1}^n \prod_{i=1}^s x^{(i)}_j$.

We provide a framework for applying the recent {\em symmetrization}
technique of Phillips \etal~\cite{PVZ12} to each of these problems.
Doing so lets us ``scale up'' each of the above lower bounds to the
$s$-player versions of the problems. 

However, the symmetrization technique
of Phillips \etal does not immediately apply, since it requires a
lower bound on the {\it distributional} communication complexity of the two-player
problem under an input distribution with certain symmetric properties.
Nevertheless, for many of the two-player lower bounds above, e.g., those in
Theorem \ref{thm:rank-lb}, our lower bound technique does not give
a distributional complexity lower bound. We instead exploit the symmetry
of the underlying problems, together with a re-randomization argument
in Theorem \ref{thm:hardest_dist} to argue that the hardest input
distribution to these problems is in fact a symmetric distribution; 
see Definition \ref{def:symmetric} for a precise definition of symmetric. We thus 
obtain a distributional lower bound by the strong version of Yao's minimax
principle. 

We obtain the following
results. Here, $\Rspubd(f)$ denotes the $\delta$-error randomized
communication complexity of the $s$-player variant of $f$. We give
precise definitions in \Cref{sec:multiparty}.

\begin{theorem} \label{thm:multiplayer}
  If $f$ is one of $\scrank_{n,n-1}$, $\scinv_n$, or $\scsls_{n,b}$,
  then $\Rs_{1/40}(f) = \Omega(sn^2\log p)$. Further,
  $\Rs_{1/12}(\scip_{n}) = \Omega(sn\log p)$ and $\Rs_{1/12}(\scham_{n,k}) = \Omega(sk)$.
\end{theorem}
We note an application to the information-theoretic 
privacy of the $\textsc{rank}_{n, n-1}$ problem in Appendix \ref{sec:privacy}.

\paragraph{Related Work:}
Many linear algebra problems have been studied in both the communication complexity model and the streaming model.
Chu and Schnitger~\cite{journals/mst/ChuS95} proved that $\Omega(n^2\log p)$ communication is required by deterministic protocols for the singularity problem over $\mathbb{F}_p$. Luo and Tsitsiklis~\cite{journals/jacm/LuoT93_1} proved that a deterministic protocol must transfer $\Omega(n^2)$ real numbers for the matrix inversion problem over $\mathbb{C}$, but Alice and Bob can only use addition, subtraction, multiplication and division of real numbers. Clarkson and Woodruff~\cite{DBLP:conf/stoc/ClarksonW09} proposed a randomized one pass streaming algorithm that uses $O(k^2\log n)$ space to decide if the rank of an integer matrix is $k$ and proved an $\Omega(k^2)$ lower bound for randomized one-way protocols in the communication complexity model via a reduction from the \textsc{Indexing} communication problem. It implies an $\Omega(n^2)$ space lower bound in the streaming model with one pass. Miltersen et al.~\cite{journals/jcss/MiltersenNSW98} showed a tight lower bound for deciding whether a vector is in a subspace of $\mathbb{F}_2^n$ in the one-sided error randomized asymmetric communication complexity model, using the Richness Lemma. Sun and Wang~\cite{sun2012randomized} proved the quantum communication complexities for matrix singularity and determinant over $\mathbb{F}_p$ are both $\Omega(n^2\log p)$.

Compared to previous results, our results are stronger. For the rank problem, the matrix singularity problem in \cite{sun2012randomized} is to decide if the rank of a matrix is $n$ or less than $n$, but \scrank$_{n,n-1}$ is to decide if the rank is $n$ or $n-1$. This additional promise enables our lower bounds for $\scinv_n$ and $\scsls_{n,b}$. If we set $k=n$ in Clarkson and Woodruff's result~\cite{DBLP:conf/stoc/ClarksonW09}, the result gives us an $\Omega(n^2)$ bound for randomized one-way protocols. However, our lower bounds work even for quantum two-way protocols. For the inverse problem, Luo and Tsitsiklis's result~\cite{journals/jacm/LuoT93_1} is in a non-standard communication complexity model, in which Alice and Bob can only make arithmetic operations on real numbers. However, our lower bound works in the standard model of communication complexity. A result of Miltersen et al.~\cite{journals/jcss/MiltersenNSW98} is to decide if a vector is in a subspace. Sun and Wang~\cite{sun2012randomized} studied the problem deciding whether two $n/2$ dimensional subspaces intersect trivially (at $\{\mathbf{0}\}$ only) or not, but we get the same bound in Corollary~\ref{cor:subspace} even with the promise. The results are analogous to the difference between set disjointness~\cite{conf/focs/BabaiFS86} and unique set disjointness~\cite{journals/siamdm/KalyanasundaramS92,razborov1992distributional}.

\begin{corollary}\label{cor:subspace}
    Alice and Bob each hold an $n/2$-dimensional subspace of $\mathbb{F}_p^n$. We promise that the intersection of the two subspaces is either $\{\mathbf{0}\}$ or a one-dimensional space. Any quantum protocol requires $\Omega(n^2 \log p)$ communication to distinguish the two cases.
\end{corollary}

In the communication model, there is another way to distribute the input: Alice and Bob each hold an $n\times n$ matrix $x$ and $y$, respectively, and they want to compute some property of $x+y$. This is equivalent to our model of matrix concatenation up to a constant factor, a fact we shall use in the paper (see~\cite{sun2012randomized} for a proof). 

\paragraph*{Paper Organization:} In Section~\ref{sec:multiparty} we present our framework of multi-party communication lower bound for a class of problems. In Section~\ref{sec:IP} we discuss the $\textsc{IP}_n$ problem and in Section~\ref{sec:rank} the $\textsc{Rank}_{n,n-1}$ problem and related linear algebra problems. 
Missing proofs, and the streaming and privacy applications are included in the Appendix. 

\section{Preliminaries}\label{sec:prelim}
{\bf Communication Complexity:}
We briefly summarize the notions from communication complexity that we will need. 
For more background on communication complexity, we refer
the reader to~\cite{kushilevitz1997communication}. 

Let $f:X \times Y \to \{1,-1\}$ be a given function, which could be a partial function. Let $\dom(f)$ be the domain of definition of $f$. Alice and Bob, with unlimited computing power, want to compute $f(x,y)$ for $(x,y)\in\dom(f)$. Alice only knows $x\in X$ and Bob $y\in Y$. To perform the computation, they follow a protocol $\Pi$ and send messages to each other in order to converge on a shared output $\Pi(x,y)$. We say a deterministic protocol $\Pi$ computes $f$ if $\Pi(x,y)=f(x,y)$ for all inputs $(x,y)\in\dom(f)$, and define the \emph{deterministic communication complexity}, denoted by $D(f)$, to be the minimum over correct deterministic protocols for $f$, of the maximum number of bits communicated over all inputs. In a randomized protocol, Alice and Bob toss private coins and the messages can depend on the coin flips. We say a randomized protocol $\Pi$ computes $f$ with error probability $\delta$ if $\Pr\{\Pi(x,y)=f(x,y)\}\geq 1-\delta$ for all inputs $(x,y)\in\dom(f)$, and define the \textit{randomized communication complexity}, denoted by $R_{\delta}(f)$, in the same way. When Alice and Bob share public random coins, the randomized communication complexity is denoted by $\R^{\pub}_{\delta}(f)$. Let $\mu$ be a probability distribution on $X\times Y$. The
$\mu$-\emph{distributional communication complexity} of $f$, denoted by $D_\delta^\mu(f)$, is the least
cost of a deterministic protocol for $f$ with error probability at most $\delta$ with
respect to $\mu$. Yao's principle states that $R_\delta^{\pub}(f) = \max_\mu D_\delta^\mu(f)$.

In the model for multiparty communication complexity, there are $s$ players, each gets an input $x_i\in X_i$, and they want to compute some function $f:X_1\times\cdots\times X_s\to \{-1, 1\}$ (which could be partially defined). We shall assume the coordinator model, in which there is an additional player called coordinator, who has no input. Players can only communicate with the coordinator but not each other directly. The coordinator will output the value of $f$. The private-coin, public-coin randomized communication complexity and $\mu$-distributional communication complexity are denoted by $R_\delta^s(f)$, $R_{\delta}^{s,\pub}(f)$, and $D_\delta^{s,\mu}(f)$, respectively.

{\bf Information Theory:} Let $(X,Y)$ be a pair of discrete random variables with joint distribution $p(x,y)$.
Suppose that $X$ is a discrete random variable on $\Omega$ with distribution $p(x)$. Then the entropy $H(X)$ of the random variable $X$ is defined by $H(X) = -\sum_{x\in\Omega} p(x)\log_2 p(x)$. The joint entropy $H(X,Y)$ of a pair of discrete random variables $(X,Y)$ with joint distribution $p(x,y)$ is defined as $H(X,Y)=-\sum_x \sum_y p(x,y)\log p(x,y)$. The conditional entropy $H(X|Y)$ is defined as $H(X|Y)= \sum_y H(X|Y=y)\Pr\{Y=y\}$, where $H(X|Y=y)$ is the entropy of the conditional distribution of $X$ given the event $\{Y=y\}$. The mutual information $I(X;Y)$ is defined as $I(X;Y)=\sum_{x,y} p(x,y)\log \frac{p(x,y)}{p(x)p(y)}$, where $p(x)$ and $p(y)$ are marginal distributions.

{\bf Information Cost:} The following two definitions are from \cite{GO13}. The \emph{information cost} $\text{ICost}(\Pi)$ of a protocol $\Pi$ on input distribution $\mu$ equals the mutual information $I(X; \Pi(X))$, where $X$ is a random variable distributed according to $\mu$
and $\Pi(X)$ is the transcript of $\Pi$ on input $X$. The information complexity $\IC_{\mu,\delta}(f)$ of a problem $f$ on a distribution $\mu$ with error probability $\delta$ is the infimum of $\text{ICost}(\Pi)$
taken over all private-randomness protocols $\Pi$ that err with probability at most $\delta$ for any input. When $\delta$ is clear from the context, we also write the information complexity as $\IC_\mu(f)$ for simplicity.

\section{Reduction for Multi-Player Communication}\label{sec:multiparty}

Let $(G,\otimes)$ be a finite group and $f$ be a function on $G$ (could be a partial function). Suppose that $G = \bigcup_i G_i$ is the coarsest partition of $G$ such that $f$ is a constant function (allowing the value to be undefined) over each $G_i$. For a subset $X\subseteq G$, let $\pre(X) := \{(g_1,g_2)\in G\times G: g_1 \otimes g_2 \in X\}$. Let $I(f) = \{i: G_i\subseteq \dom(f) \}$, where $\dom(f)\subseteq G$ is the set on which $f$ is defined.

We say that a family $\mathcal{H}$ of functions $h:G\times G\to G\times G$ is a \emph{uniformizing family} for function $f$ if there exists a probability measure $\mu$ on $\mathcal{H}$ such that for any $i$ and $(g_1,g_2)\in \pre(G_i)$, when $h\in\mathcal{H}$ is randomly chosen according to $\mu$, the image $h(g_1,g_2)$ is uniformly distributed on $\pre(G_i)$. 

\begin{example}[$\textsc{Rank}_{n,n-1}$] $G = M_n(\mathbb{F})$, the group of all $n\times n$ matrices over $\F$, with $\otimes$ being the usual matrix addition. In fact $G$ is a ring, with the usual matrix multiplication. Define
\[
f(x) = \begin{cases}
		1, & \rk(x) = n;\\
		0, & \rk(x) = n-1;\\
		\text{undefined}, & \text{otherwise},
\end{cases}\qquad x\in G.
\]
Then $I(f) = \{1,2\}$ and $G_1 = \{x\in G: \rk(x) = n\}$ and $G_2 = \{x\in G: \rk(x) = n-1\}$.
The uniformizing family is $\mathcal{H} = \{h_{a,b}\}_{a\in G_1,b\in G}$ endowed with uniform measure, where $h_{a,b}(g_1,g_2) = (a(g_1-b),a(g_2+b))$.
\end{example}

\begin{example}[$\textsc{Ham}_{n,k}$] $G = \mathbb{F}_2^n$ with the usual vector addition. Define
\[
f(x) = \begin{cases}
		1, & w(x) = k;\\
		0, & w(x) = k+2;\\
		\text{undefined}, & \text{otherwise},
\end{cases}\qquad x\in G.
\]
Then $|I(f)|=2$.
Let $S_n$ denote the symmetric group of degree $n$. The uniformizing family $\mathcal{H} = \{h_{\sigma,b}\}_{\sigma\in S_n,b\in G}$ endowed with uniform measure, where $h_{\sigma,b}(g_1,g_2) = (\sigma(g_1-b),\sigma(g_2+b))$.

By reduction from Disjointness problem, we know that $\R_{1/10}^{\pub}(\textsc{Ham}_{k,k+2}) = \Omega(k)$.
\end{example}

As an auxiliary problem to the $\textsc{IP}$ problem, we define
\begin{itemize}
\item Problem $\textsc{IP}_n'$: Suppose that $p > 2$. Alice and Bob hold two vectors $x,y\in (\F_p^\ast)^n$ respectively. We promise that inner product $\langle x,y\rangle\in\{0,1\}$. They want to output $\langle x,y\rangle$.
\end{itemize}

Removing $0$ from the scalar domain gives us a group structure as below.

\begin{example}[$\textsc{IP}'_n$] 
$G = (\mathbb{F}_p^\ast)^n$ associated with the multiplication $\otimes$ defined to be the pointwise product, i.e., $x\otimes y = (x_1y_1,x_2y_2,\dots,x_ny_n)$. Let $f(x) = \mathbf{1}_{\{x_1+x_2+\cdots+x_n=0\}}$.
\end{example}

The following problem was considered in \cite{SWY12}.
\begin{itemize}
\item Problem $\textsc{Cycle}_{n}$: Let $\pi$ and $\sigma$ be permutations in symmetric group $S_n$. Alice holds $\pi$ and Bob $\sigma$, and they want to return $1$ if $\pi\circ\sigma$ is exactly 1-cycle and return $0$ otherwise.
\end{itemize}

\begin{example}[$\textsc{Cycle}_{n}$] $G = S_n$, the symmetric group of degree $n$, with the usual permutation composition. Define
\[
f(x) = \begin{cases}
		1, & x\text{ has exactly one cycle};\\
		0, 
		& \text{otherwise},
		\end{cases} \qquad x\in G.
\]
Then $|I(f)|=2$. The uniformizing family is $\mathcal{H} = \{h_{\sigma,\tau}\}_{\sigma,\tau\in S_n}$ endowed with uniform measure, where $h_{\sigma,\tau}(g_1,g_2) = (\sigma^{-1}g_1\tau^{-1},\tau g_2\sigma)$. Observe that $g\mapsto \sigma^{-1} g\sigma$ maps a cycle $(a_1,\dots,a_k)$ of $g$ to $(\sigma(a_1),\dots,\sigma(a_k))$, it is easy to verify that $\mathcal{H}$ is a uniformizing family indeed. It has been shown in \cite{SWY12} that $\R_{1/3}^\pub(\textsc{Cycle}_{n}) = \Omega(n)$.
\end{example}
We analyze the randomized communication complexity of problems that have a uniformizing family.
\begin{definition}\label{def:symmetric} A distribution $\nu$ on $G\times G$ is called \emph{weakly sub-uniform} if
\begin{enumerate}
	\item $\nu$ is supported on $\bigcup_{i\in I(f)} \pre(G_i)$ 
	\item $\nu|_{\pre(G_i)}$ is uniform for all $i\in I(f)$
\end{enumerate}
In addition, if $\nu(\pre(G_i)) = 1/|I(f)|$ for all $i\in I(f)$, 
we say $\nu$ is the \emph{sub-uniform} distribution.
\end{definition}

\begin{theorem}\label{thm:hardest_dist} If there exists a uniformizing family for $f$ and $\delta\cdot |I(f)| < 1$, then for the two-player game computing $f$ it holds that 
\[
\R_{\delta|I(f)|}^{\pub}(f) \leq D_{\delta}^{\nu}(f) \leq C\log_{|I(f)|\delta}\delta\cdot R^{\pub}_{|I(f)|\delta}(f)
\]
where $C>0$ is an absolute constant and $\nu$ the sub-uniform distribution on $G\times G$.
\end{theorem}
\begin{proof}
Suppose the input is $(g_1,g_2)\in G\times G$. Next we describe a public-coin protocol $\Pi'$. With the public randomness, Alice and Bob choose a random $h$ from the uniformizing family. They then run the optimal protocol $\Pi_{\nu}$ for input distribution $\nu$ (i.e., $\cost(\Pi_{\nu}) = D_\delta^{\nu}(f)$) on input $h(g_1,g_2)$. 

It is not difficult to see that the public-coin protocol $\Pi'$ has error probability at most $\delta\cdot |I(f)|$. Therefore, $R^{\pub}_{\delta\cdot |I(f)|}\leq \cost(\Pi') = \cost(\Pi_\nu) = D_{\delta}^{\nu}(f)$. On the other hand, by Yao's principle, $\R_\delta^{\pub}(f)\geq D_\delta^{\nu}(f)$. Note that $R^{\pub}_\delta(f)\leq C \log_{|I(f)|\delta}\delta\cdot R^{\pub}_{|I(f)|\delta}(f)$ for some absolute constant $C$, the conclusion follows.
\end{proof}

Now consider the following multi-player problem in coordinator model: There are $s$ players and a coordinator. Each player receives an input $x_i\in G$. The coordinator will output the value of $f(x_1\otimes x_2\otimes \cdots\otimes x_s)$ with probability $\geq 1-\delta$. Denote by $C_\delta^{s,\pub}(f)$ the number of bits that must be exchanged by the best protocol. By the symmetrization technique from \cite{PVZ12}, we have the following lemma.

\begin{lemma} Suppose that there exists a uniformizing family for $f$. Let $\nu$ be an arbitrary weakly sub-uniform distribution on $G\times G$ and $\Pi_\nu$ be a public-coin protocol that computes $f$ with error probability $\delta$ on input distribution $\nu$. Then $\R_{\delta}^{s,\pub}(f) \geq s\mathbb{E}[\cost(\Pi_\nu)]$.
\end{lemma}
\begin{proof}
Let $\nu_s$ be the distribution over $G^s$ such that $\nu_s$ is the uniform distribution over $\pre_s(G_i):=\{(x_1,\dots,x_s)\in G^s: x_1\otimes\cdots\otimes x_s\in G_i\}$ when restricted onto it and $\nu_s(\pre_s(G_i)) = \nu(\pre(G_i))$. Let $\Pi_s$ be an $s$-player (deterministic) protocol for input distribution $\nu_s$ with error probability $\delta$.

Consider the following two-player protocol $\Pi'$ on input $(g_1,g_2)\sim \nu$: First suppose that Alice and Bob have public randomness. They first use the public randomness to agree on an index $j$ chosen at random uniformly from $\{1,\dots,s\}$. Alice also generates, using her own randomness, the input $\{x_i\}_{i\neq j}$ of other players uniformly at random conditioned on $\bigotimes_{i\neq j} x_i = g_1$. Then Alice and Bob run the $s$-player protocol, in which Bob simulates player $j$ with input $x_j := g_2$, and Alice simulates all other players and the coordinator. The message sent in this protocol is just the message sent between the coordinator and player $j$ in $\Pi_s$. 

It is not hard to see that $(x_1,\dots,x_s)\sim \nu_s$. It follows from a symmetrization argument like the proof \cite[Theorem 1.1]{PVZ12} that $\mathbb{E}[\cost(\Pi')] \leq \cost(\Pi_s)/s$, where the expectation is taken over the public coins. The conclusion follows from taking the infimum over $\Pi_s$.
\end{proof}

\begin{theorem}\label{thm:multiparty}
Suppose that there exists a uniformizing family for $f$, then $\R_{\delta}^{s,\pub}(f) \geq \delta s \R_{2|I(f)|\delta}^{\pub}(f)$.
\end{theorem}
\begin{proof}
Pick $\nu$ to be the sub-uniform distribution in the preceding lemma. By fixing the public coins and a Markov bound, one can construct a two-player deterministic protocol $\Pi''$ such that $\cost(\Pi'') \leq (1/\delta)\cost(\Pi_s)/s$ and $\Pi''$ succeeds with probability at least $1-2\delta$ when the input is distributed as $\nu$. Hence $D_{2\delta}^\nu(f) \leq (1/\delta)\cost(\Pi_s)/s$. It then follows from Theorem~\ref{thm:hardest_dist} that 
$R^{\pub}_{2|I(f)|\delta}(f) \leq (1/\delta)\cost(\Pi_s)/s$. Taking infimum over $\Pi_s$, we obtain that $\R_{2|I(f)|\delta}^{\pub}(f) \leq (1/\delta) D_{\delta}^{\nu_s}(f)/s \leq (1/\delta) R^{s,\pub}_{\delta}(f)/s$.
\end{proof}

The following are immediate corollaries of the theorem above applied to our previous Example 2 and 4. We leave the results of Example 1 and 3 for later sections.
\begin{corollary} $\R_{1/12}^{s,\pub}(\textsc{Ham}_{k,k+2}) = \Omega(s\R_{1/3}^{\pub}(\textsc{Ham}_{k,k+2})) = \Omega(sk)$.
\end{corollary}

\begin{corollary} $\R_{1/12}^{s,\pub}(\textsc{Cycle}_{n}) = \Omega(sn)$. 
\end{corollary}

\section{The \textsc{IP} Problem}\label{sec:IP}

Let $p$ be a prime. Sun \etal considered a variant of the \textsc{IP} problem, denoted by $\textsc{IP}_n''$, in which Alice has $x\in F_p^n$ and Bob $y\in (F_p^\ast)^n$, and showed that $\R_{1/3}^{\pub}(\textsc{IP}_n'') = \Omega(n\log p)$ \cite{SWY12}. Via a simple reduction, we show that

\begin{lemma}\label{lem:IP_reduction} When $p\geq p_0$ for some constant $p_0$, it holds that $\R_{1/3}^{\pub}(\textsc{IP}'_n) = \Omega(n\log p)$.
\end{lemma}

\begin{proof}
For an input of $\textsc{IP}''$, Alice can send the indices of the zero coordinates to Bob using $n$ bits; on the remaining coordinates, Alice and Bob have an instance of $\textsc{IP}'$ of size at most $n$. Hence $n + \R_{1/3}^{\pub}(\textsc{IP}'_n)\geq \R_{1/3}^{\pub}(\textsc{IP}_n'')$, whence the conclusion follows.
\end{proof}

It is clear, by Yao's principle, that $\R_\delta^{\pub}(\textsc{IP}_n)\geq \R_\delta^{\pub}(\textsc{IP}'_n)$.
Now, as an immediate corollary of Theorem~\ref{thm:multiparty}, we have 
\begin{theorem} $\R_{1/12}^{s,\pub}(\textsc{IP}_n) = \Omega(sn\log p)$.
\end{theorem}
\begin{proof}
Let $p_0$ be as in Lemma~\ref{lem:IP_reduction}. It follows from Lemma~\ref{lem:IP_reduction} and Theorem~\ref{thm:multiparty} that $
\R_{1/12}^{s,\pub}(\textsc{IP}) \geq \R_{1/12}^{s,\pub}(\textsc{IP}') = \Omega(s\R_{1/3}^{\pub}(\textsc{IP}')) = \Omega(sn\log p)$. 
When $p < p_0$, the result is due to Braverman et al. in~\cite{BEOPV13}, who prove an $\Omega(sn)$ lower bound
for $\textsc{IP}$ over the integers with the promise that the inner product is $0$ or $1$. Note that this implies an $\Omega(sn \log p)$
lower bound for computing $\textsc{IP}$ over $\mathbb{F}_p$ as well, since $p < p_0$ is a fixed constant. 
\end{proof}

\section{The $\scrank$ Problem}
\label{sec:rank}

We shall use the Fourier witness method to prove a lower bound on \scrank$'_{n,n-1}$. We then use this result for \scrank$_{n,n-1}$ to obtain lower bounds for the other problems. We review some basics of the Fourier witness method in Section~\ref{sec:fourierwitness} then give the proof of the lower bound in Section~\ref{sec:rank_proof}.

\subsection{Fourier Witness Method}
\label{sec:fourierwitness}

%
%
%

\subsubsection{Fourier Analysis}
For prime $p$, let $\mathbb{F}_p$ be the finite field of $p$ elements. We define the Fourier transformation on the group $(\mathbb{F}_p^n,+)$.
\begin{definition}[Fourier transform]
    Let $f:\mathbb{F}_p^N \to \mathbb{R}$ be a function. Then, the Fourier coefficient of $f$, denoted by $\hat{f}$, is also a $\mathbb{F}_p^N \to \mathbb{R}$ function, defined as
    $\hat{f}(s)=\frac{1}{p^N}\sum_{x\in \mathbb{F}_p^N} \omega^{-\langle s,x\rangle} f(x),$
    where $\omega=e^{2\pi i/p}$. 
\end{definition}


\begin{fact} \label{pro:inverse_fourier}  $ f = p^{N}\Big(\widehat{\big(\hat{f}\big)^\ast}\Big)^\ast $.
\end{fact}

\subsubsection{Approximate Norm and Dual Norm}

The $\ell_p$ norm of a vector $v\in \mathbb{R}^n$ is defined by $\|v\|_p:= \left( \sum_{i=1}^n |v_i|^p \right) ^ {1/p}$ and the $\ell_\infty$ norm by $\|v\|_\infty:= \max_{i=1}^n |v_i|$. The trace norm of an $n\times n$ matrix $F$, denoted by $\|F\|_\tr$, is defined as $\|F\|_\tr := \sum_i \sigma_i $, where $\sigma_1,\cdots,\sigma_n$ are the singular values of $F$.

The matrix rank and some matrix norms can give lower bounds for deterministic communication complexity. For randomized lower bounds, we need the notions of approximate rank and norms.
\begin{definition} [approximate norm]
    Let $\rho: \mathbb{R}^X \mapsto \mathbb{R}$ be an arbitrary norm and $f: X\mapsto \mathbb{R}$ a partial sign function. The $\epsilon$-approximate $\rho$ norm of $f$, denoted by $\rho^\epsilon(f)$, is defined as
    $ \rho^\epsilon(f) = \inf_\phi \rho(\phi) $,
where the infimum is taken over all functions $\phi:X\mapsto \mathbb{R}$ that satisfy
    $$  \phi(x) \in
        \begin{cases}
            [1-\epsilon, 1+\epsilon] & \text{if $f(x)=1$;}  \\
            [-1-\epsilon, -1+\epsilon] & \text{if $f(x)=-1$;} \\
            [-1-\epsilon, 1+\epsilon] & \text{if $f(x)$ is undefined.} \\
        \end{cases}
    $$
\end{definition}

The following lemma shows that the approximate trace norm gives lower bounds on quantum communication complexity, as well as on randomized protocols with public coins. The following lemma is a result in \cite{lee2009lower} combined with Neumann's argument for converting a public-coin protocol into a private-coin one.

\begin{lemma} \label{cor:approx_trace_norm_method}
For $\delta > 0$ such that $1/(1-2\epsilon) \leq 1+ \delta$, it holds that
\[
\R^{\pub}_\delta(f) \geq \Omega\left( \log \frac{(\|F\|_\tr^\epsilon)^2}{\size(F)} \right) - O(\log n + \log\frac{1}{\delta}).
\]
\end{lemma}

The approximate norms are minimization problems. We will consider the dual problems, which are maximization problems.

\begin{definition}
    Let $\rho$ be an arbitrary norm on $\mathbb R^n$. The dual norm of $\rho$, denoted by $\rho^*$, is defined as
    $$ \rho^*(v) = \sup_{u:\rho(u)\leq 1} \langle v,u\rangle. $$
\end{definition}

The following lemma characterizes the approximate norm as a maximization problem so that we can prove \emph{lower} bounds more easily.
\begin{lemma}[\cite{sherstov2012strong}] \label{lem:dual_and_approx}
    Let $f$ be a partial sign function and $\rho$ an arbitrary norm. Then
    $$ \rho^\epsilon(f) = \sup_{\psi\neq 0} \frac{\langle f, \psi \cdot \dom(f) \rangle - \|\psi \cdot \overline{\dom(f)} \|_1 - \epsilon \|\psi\|_1}{\rho^*(\psi)},\quad \epsilon > 0. $$
where
    $$ \dom(f)(x) =
        \begin{cases}
                     1 & \text{if $f(x)$ is defined,} \\
                     0 & \text{otherwise,}
        \end{cases} $$
$\overline{\dom(f)}(x)=1-\dom(f)(x)$, and $(\psi\cdot\varphi)(x)=\psi(x)\varphi(x)$.
\end{lemma}


We call a feasible solution in the dual problem the witness of the original problem. In particular, in Lemma~\ref{lem:dual_and_approx}, the function $\psi$ is the witness. Any $\psi$ gives a lower bound for $\rho^\epsilon(f)$. It is difficult to find a useful witness. The first choice that comes to mind is to choose $\psi=f\cdot \dom(f)$, because it makes $\langle f, \psi \cdot \dom(f) \rangle$ large and  $\|\overline{\dom(f)} \|_1$ small. This is the discrepancy method. We use a different choice: $\psi= \widehat{\big(f\cdot \dom(f)\big)}$. We call it the \emph{Fourier witness method}, introduced in \cite{sun2012randomized}, but used here for partial functions. 

\begin{definition} [approximate Fourier $p$-norm]
    Let $f:\mathbb{F}_p^N\mapsto \mathbb R$ be a function and $p\geq 1$. The Fourier $p$-norm of $f$, denoted by $\|\hat{f}\|_p$, is the $p$-norm of $\hat{f}$. Furthermore, if $f$ is a sign function, the approximate Fourier $p$-norm of $f$, denoted by $\|\hat{f}\|_p^\epsilon$, is the approximate $\|\hat{\cdot}\|_p$ norm of $f$.
\end{definition}


\begin{fact} \label{lem:dual_fourier_l1}
    The dual norm of $\|\cdot\|_1$ is $\|\cdot\|_\infty$.
    The dual norm of $\|\hat{\cdot}\|_1$ is $p^N\|\hat{\cdot}\|_\infty$.
\end{fact}

The Fourier coefficients of a plus composed function are related to the singular values of the associated matrix, as shown by the following lemma, whose proof is postponed to Appendix~\ref{sec:eigenvalue_and_fourier}. Applied to approximate trace norm, it also builds a bridge between the approximate trace norm and the approximate Fourier $\ell_1$-norm for a plus composed function. Similar results and additional background can be found in \cite{lee2009lower,sun2012randomized}.

\begin{lemma}\label{lem:eigenvalue_and_fourier}
Suppose that $g:\mathbb{F}_p^N\mapsto \mathbb{R}$ is a function, and $f$ is a plus-composed function $f(x,y)=g(x+y)$. Let $F$ be the associated matrix of $f$. Then the singular values of $F$ are $p^N$ times the modulus of the Fourier coefficients of $g$, i.e.
$ \sigma_F = p^N\cdot |\hat{g}|$,
where $\sigma_F$ are the singular values of $F$ and $|\hat{g}|(s)=|\hat{g}(s)|$. As a consequence,  $\|F\|_\tr^\epsilon = p^N \cdot \|\hat{g}\|_1^\epsilon$.
\end{lemma}

\subsection{\scrank$_{n,n-1}$}\label{sec:rank_proof}
For a matrix $x\in \mathbb{F}_p^{n\times n}$, we define $\theta(x)=1$ if $x$ is of full rank and $\theta(x)=0$ otherwise. We shall use $\theta$ as the witness in the proof of \scrank$'_{n,n-1}$. The same function $\theta$ has been used to prove a communication complexity lower bound for the matrix singularity problem in~\cite{sun2012randomized}. 
\begin{theorem}\label{thm:rank_lb}
$\R^{\pub}_{1/10}(\scrank_{n,n-1})=\Omega(n^2 \log p)$.
\end{theorem}
\begin{proof}
Suppose that $\Pi$ is a public-coin protocol for $\scrank'_{n,n-1}$ with error probability $\leq 1/10$. Then Alice and Bob can build a public-coin protocol $\Pi'$ as follows. They use the public coins to choose a random matrix $r$ and run $\Pi$ on input $(x-r,y+r)$. It is easy to see that $\Pi'$ has error probability $\leq 1/10$ and $\cost(\Pi')=\cost(\Pi)$. Observe that the distribution of $\Pi'(x,y)$ is identical to the that of $\Pi'(a,b)$ whenever $x+y=a+b$. 

Define the partial sign function 
$g(x) = 1$ if $\rankp(x) = n$, $g(x) = -1$ if $\rankp(x) = n-1$, and
$g(x)$ is undefined otherwise.
Let 
$f(x,y)$ be the expected output of $\Pi'(x,y)$. Then $f$ is a plus-composed function. By the correctness of $\Pi$, we know that $f(x,y)=g(x+y)$ whenever $g(x+y)$ is defined. We claim that $\|g\|_1^\epsilon = \Omega(p^{n(n-3)/2})$ for $\epsilon = 1/4$, following  Lemma~\ref{lem:dual_and_approx} (applied with witness $\theta$ as in the paragraph before the theorem statement) and Fact~\ref{lem:dual_fourier_l1}. See Appendix~\ref{app:proofs} for details. 
Finally, it follows from Lemma~\ref{cor:approx_trace_norm_method} that
\begin{multline*}
\R^{\pub}_{1/10}(f)
=\Omega\left( \log \frac{\|F\|_\tr^{1/4}}{\sqrt{\size(F)}} \right) - O(\log n)
=\Omega\left( \log \frac{p^{n^2}\|\hat{g}\|_1^\epsilon}{\sqrt{\size(F)}} \right)- O(\log n)\\
=\Omega\left( \log \frac{p^{n^2}\cdot 0.4 p^{n(n-3)/2}}{p^{n^2}} \right)  - O(\log n)
=\Omega(n^2\log p).\qedhere
\end{multline*}
\end{proof}

The lower bound for the multi-player $\textsc{Rank}$ problem is an immediate corollary of Theorem~\ref{thm:multiparty}.
\begin{corollary} $R_{1/40}^{s,\pub}(\textsc{Rank}_{n,n-1}) = \Omega(sR_{1/10}^{\pub}(\textsc{Rank}_{n,n-1})) = \Omega(s n^2\log p)$.
\end{corollary}

By padding zeros outside the top-left $k\times k$ submatrix, we obtain a lower bound for $\scrank_{k,k-1}$.
\begin{corollary} $\R^{\pub}_{1/10}(\scrank_{k,k-1})=\Omega(k^2 \log p)$.
\end{corollary}

\subsection{Linear Algebra Problems}
\label{sec:la}
\begin{problem}[\scinv]
Alice and Bob hold two $n\times n$ matrices $x$ and $y$ over $\mathbb{F}_p$, respectively. We promise that $x+y$ is invertible over $\mathbb{F}_p$. They want to determine if the top-left entry of $(x+y)^{-1}$ is zero (output $-1$) or non-zero (output $1$).
\end{problem}

\begin{problem}[\scsls]
Alice and Bob hold two $n\times n$ matrices $x$ and $y$ over $\mathbb{F}_p$, respectively. We promise that $x+y$ is invertible over $\mathbb{F}_p$. $b$ is a parameter of this problem. $t$ is the vector of variables of the linear system $(x+y)t=b$. They want to determine if the first coordinate of $t$ is zero.
\end{problem}
%
%

\begin{theorem}
$\R^{\pub}_{1/20}(\scinv) = \Omega(n^2\log p)$ for $p\geq 3$.
\end{theorem}
\begin{proof}
We reduce $\scrank$ to $\scinv$. Let $A = x+y$ and $\tilde A$ be the lower-right $(n-1)\times (n-1)$ block of $A$. Then $A^{-1}_{11} = 0$ iff $\rk(\tilde A) < n-1$. 

Now, suppose that $A$ is an $(n-1)\times (n-1)$ matrix and $\rk(A)\in\{n-1,n-2\}$. We augment $A$ to $A_1$ by appending a random column. With probability $1-1/p$ it holds that $\rk(A_1) = n-1$ when $\rk(A) = n-2$. Now we augment $A_1$ to $A_2$ by appending a random row. With probability $1-1/p$ it holds that $\rk(A_2)=n$ when $\rk(A_1) = n-1$.

Run a protocol for $\scinv$ on $A_2$. We denote the communication complexity of the protocol by $c(n)$. When $\rk(A) = n-1$, if the error probability of the protocol is at most $1/20$, then it outputs $1$ with probability 
$
\alpha \leq \frac{1}{20}\big(1-\frac 1p\big) + \frac 1p,$ 
while when $\rk(A)=n-2$ it outputs $1$ with probability
$\beta \geq \frac{19}{20}\big(1-\frac 1p\big)^2.$
Then
$\beta-\alpha \geq \frac{19}{20}\big(1-\frac 1p\big)^2-\frac{1}{20}\big(1-\frac 1p\big)-\frac{1}{p}\geq \frac{1}{18},\quad p\geq 3,$
which implies that $\Theta(1)$ independent repetitions allow us to solve \scrank\ on $(n-1)\times (n-1)$ matrices, i.e., to distinguish $\rk(A)=n-1$ from $\rk(A)=n-2$, with error probability $\leq 1/20$ and communication complexity $\Theta(c(n)) = \Omega((n-1)^2\log p)$. Therefore $c(n) = \Omega((n-1)^2\log p) = \Omega(n^2\log p)$.
\end{proof}

\begin{theorem}
$\R^{\pub}_{1/20}(\scinv) = \Omega(n^2)$ for $p=2$.
\end{theorem}
\begin{proof}
As before, we augment $A$ to $A_2$. Here we further randomize $A_2$ by multiplying a random invertible matrix on both sides of $A_2$, that is, we form $B=G_1 A_2 G_2$ where $G_1,G_2$ are uniform over $n\times n$ non-singular matrices over $\F_p$. It is clear that $\rk(B)=\rk(A_2)$, and $B$ is uniformly distributed over the $n\times n$ matrices with the same rank.

Run a protocol for $\scinv$ on $B$. Suppose that it outputs zero with probability $p_0$ when the input matrix has rank $n-1$. This probability can be calculated by Alice and Bob individually with no communication cost.
When $\rk(A) = n-1$, it outputs $1$ with probability 
$\alpha = \frac{1}{20}\big(1-\frac 1p\big) + \frac{p_0}{p} = \frac{1}{40} + \frac{p_0}2,$
while when $\rk(A)=n-2$ it outputs $1$ with probability
$\beta \geq \frac{19}{20}\big(1-\frac 1q\big)^2 + p_0\cdot \frac{2}{p}\big(1-\frac 1p\big) = \frac{19}{80}+\frac{p_0}{2}.$
Then,
$\alpha-\beta \geq \frac{17}{80}.$
The rest follows as in the proof for $p\geq 3$. 
\end{proof}

Now we reduce \scinv\ to \scsls$_b$.
\begin{theorem} $\R^{\pub}_{1/20}(\scsls_b)=\Omega(n^2 \log p)$ for $b\neq \mathbf 0$,
\end{theorem}
\begin{proof}
    We prove it by a reduction from $\scinv$.

    Take an instance $(x,y)$ from $\scinv$. Since $b\neq 0$, there exists an invertible matrix $Q$ such that $Qb=(1,0,0,\cdots,0)\trans$. Alice and Bob agree on the same $Q$, e.g. the minimal $Q$ in alphabetical order. Then they run the protocol of \scsls\ on input $(Q^{-1}x,Q^{-1}y,b)$.  Then $
t=(Q^{-1}x+Q^{-1}y)^{-1}b=(x+y)^{-1}QQ^{-1}(1,0,\cdots,0)\trans=(x+y)^{-1}(1,0,\cdots,0)\trans$
and thus $t_1=((x+y)^{-1})_{11}$.
\end{proof}

\paragraph{Acknowledgements}
We would like to thank Amit Chakrabarti for reading and helping with an earlier draft of this paper, and Troy Lee for giving us useful suggestions. David Woodruff would also like to acknowledge the XDATA program of the Defense Advanced Research Projects Agency (DARPA), administered through Air Force Research Laboratory contract FA8750-12-C0323, for support for this project. 
\bibliographystyle{plain}
\bibliography{matrixinv}

\newpage
\appendix
\section{Proof of Lemma~\ref{lem:eigenvalue_and_fourier}}\label{sec:eigenvalue_and_fourier}
\begin{proof}
    Let $\omega=e^{2\pi/p}$, $U(x,y)=p^{-N/2}\cdot \omega^{-\langle x,y\rangle}$, and $\Lambda = \diag(|\hat{g}|^2)$. We will prove $F^\dag F=p^N U \Lambda U^\dag$.
    \begin{eqnarray*}
            (F^\dag F)_{x,z}
        &=& \sum_y (F_{y,x})^* F_{y,z} \\
        &=& \sum_y g(y+x)^* g(y+z) \\
        &=& \sum_y \sum_s \omega^{-\langle s,y+x\rangle} \hat{g}(s)^* \sum_t \omega^{\langle t,y+z\rangle} \hat{g}(t) \\
        &=& \sum_s \sum_t  \omega^{-\langle s,x\rangle+\langle t,z\rangle} \hat{g}(s)^* \hat{g}(t) \left( \sum_y \omega^{\langle -s+t,y\rangle} \right) \\
        &=& \sum_s \sum_t  \omega^{-\langle s,x\rangle+\langle t,z\rangle} \hat{g}(s)^* \hat{g}(t) \cdot p^N\delta_{s,t} \\
        &=& p^N \cdot \sum_s \omega^{-\langle s,x\rangle+\langle s,z\rangle} \hat{g}(s)^* \hat{g}(s) \\
        &=& p^N \cdot \sum_s \omega^{-\langle s,x\rangle} |\hat{g}(s)|^2 \omega^{\langle s,z\rangle}  \\
        &=& p^{2N} \cdot \sum_s U_{x,s} |\hat{g}(s)|^2 (U_{z,s})^* \\
        &=& p^{2N}(U\Lambda U^\dag)_{x,z}
    \end{eqnarray*}
    $U$ is unitary, because
    \begin{eqnarray*}
            (U^\dag U)_{x,z}
        &=& \sum_y (U_{y,x})^* U_{y,z} \\
        &=& \sum_y p^{-N/2} \omega^{\langle y,x\rangle} \cdot p^{-N/2} \omega^{-\langle y,z\rangle} \\
        &=& p^{-N}\cdot \sum_y \omega^{\langle y,x-z\rangle} \\
        &=& p^{-N}\cdot p^N \delta_{x,z} \\
        &=& \delta_{x,z}
    \end{eqnarray*}
Therefore, the singular values of $F$ are $p^N\cdot |\hat{g}|$.
\end{proof}

\section{Proofs for the $\scrank_{n,n-1}$ Problem}
\label{app:proofs}

We shall use the following fact.
\begin{fact}[\cite{morrison2006integer}]
	The number $n\times n$ matrices over $\mathbb{F}_p$ of rank-$r$ is
\[	\frac{(p^n-1)(p^n-p)\cdots(p^n-p^{r-1})}{(p^r-1)(p^r-p)\cdots(p^r-p^{r-1})} \prod_{k=0}^{r-1}(p^n-p^k). \]
\end{fact}

The following lemma computes the Fourier coefficients and shows that the $\ell_1$-norm of $\hat{\theta}$ is small.
\begin{lemma} \label{lem:fouriercoef}
Let $r=\rankp(s)$, then
    \[\hat{\theta}(s)=(-1)^r p^{-n(n+1)/2} \prod_{k=1}^{n-r}(p^k-1).\]
Hence
    \[\|\hat{\theta}\|_1 = p^{-n} \prod_{k=1}^n (p^k-1) \prod_{k=0}^{n-1} \frac{1+p^k}{p^k}.\]
\end{lemma}
\begin{proof}
The expression of $\hat\theta(s)$ is from \cite{sun2012randomized}. It follows straightforwardly that 
\begin{align*}
 \|\hat{\theta}\|_1
&= \sum_{r=0}^n \sum_{s:\rankp(s)=r} |\hat{\theta}(s)| \\
&= \sum_{r=0}^n p^{r(r-1)/2} \binom{n}{r}_{\!\!p}\
 \prod_{k=n-r+1}^n (p^k-1) \cdot p^{-n(n+1)/2} \prod_{k=1}^{n-r}(p^k-1) \\
&= p^{-n(n+1)/2} \prod_{k=1}^n (p^k-1) \sum_{r=0}^n p^{r(r-1)/2} \binom{n}{r}_{\!\!p} \\
&= p^{-n(n+1)/2} \prod_{k=1}^n (p^k-1) \prod_{k=0}^{n-1} (1+p^k) \\
&= p^{-n} \prod_{k=1}^n (p^k-1) \prod_{k=0}^{n-1} \frac{1+p^k}{p^k} 
\end{align*}
\end{proof}

We are now ready to complete the proof of Theorem~\ref{thm:rank_lb}.

\begin{proof}
We shall show that $\|\hat{g}\|_1^\epsilon$ is large. By Lemma~\ref{lem:dual_and_approx} and Fact~\ref{lem:dual_fourier_l1},
$$ \|\hat{g}\|_1^\epsilon \geq \frac{\langle g, \psi \cdot \dom(g) \rangle - \|\psi \cdot \overline{\dom(g)} \|_1 - \epsilon \|\psi\|_1}{p^{n^2}\|\hat{\psi}\|_\infty}.$$
Choosing $\psi=(-1)^n \hat{\theta}$ yields that
\begin{align*}
&\quad\  \langle g,\psi \cdot \dom(g) \rangle \\
&= \sum_{x:\rankp(x)=n} (-1)^n\hat{\theta}(x) - \sum_{x:\rankp(x)=n-1} (-1)^n\hat{\theta}(x) \\
&= \sum_{x:\rankp(x)=n} |\hat{\theta}(x)| + \sum_{x:\rankp(x)=n-1} |\hat{\theta}(x)| \\
&= p^{n(n-1)/2} \prod_{k=1}^{n}(p^k-1) \cdot p^{-n(n+1)/2} + \frac{p^n-1}{(p-1)^2} p^{(n-1)(n-2)/2} \prod_{k=1}^{n}(p^k-1) \cdot p^{-n(n+1)/2} (p-1) \\
&= \left( 1 + \frac{p-p^{-n+1}}{p-1} \right) p^{-n} \prod_{k=1}^{n}(p^k-1).
\end{align*}
Observe that
\begin{gather*}
\|\psi \cdot \dom(g) \|_1
= \sum_{x:\rankp(x)=n} |\hat{\theta}(x)| + \sum_{x:\rankp(x)=n-1} |\hat{\theta}(x)|
= \langle g, \psi \cdot \dom(g) \rangle\\
     \|\psi \cdot \overline{\dom(g)} \|_1
 = \|\psi\|_1 - \|\psi \cdot \dom(g)\|_1 
 = \|\hat{\theta}\|_1 - \|\psi\cdot \dom(g)\|_1 \\
 \|\psi\|_1 = \|\hat\theta\|_1\\
 \|\hat\psi\|_\infty = \|\hat{\hat\theta}\|_\infty =  p^{-n^2}.
\end{gather*}
The lower bound for $\|\hat{g}\|_1^\epsilon$ follows as below.
\begin{align*}
      \|\hat{g}\|_1^\epsilon 
&\geq \frac{\langle g, \psi \cdot \dom(g) \rangle - \|\psi \cdot \overline{\dom(g)} \|_1 - \epsilon \|\psi\|_1}{p^{n^2}\|\hat{\psi}\|_\infty} \\
& =   \frac{\langle g, \psi \cdot \dom(g) \rangle - (\|\hat{\theta}\|_1 - \langle g, \psi \cdot \dom(g) \rangle)- \epsilon \|\hat{\theta}\|_1}{p^{n^2}\cdot p^{-n^2}} \\
& =   2 \langle g, \psi \cdot \dom(g) \rangle - (1+\epsilon)\|\hat{\theta}\|_1 \\
& =   2\cdot \left( 1 + \frac{p-p^{-n+1}}{p-1} \right) p^{-n} \prod_{k=1}^{n}(p^k-1) -(1+\epsilon)\cdot p^{-n} \prod_{k=1}^n (p^k-1)\prod_{k=0}^{n-1} \frac{1+p^k}{p^k} \\
& =   \left(2\left( 1 + \frac{p-p^{-n+1}}{p-1} \right)-(1+\epsilon)\prod_{k=0}^{n-1} \frac{1+p^k}{p^k} \right)p^{-n} \prod_{k=1}^{n}(p^k-1) \\
\end{align*}
Note that
$$
       p^{-n} \prod_{k=1}^{n}(p^k-1) 
\geq p^{-n} \prod_{k=1}^{n}p^{k-1}
=    p^{n(n-3)/2},
$$
and when $\epsilon=1/4$,
\begin{gather*}
2\left( 1 + \frac{p-p^{-n+1}}{p-1} \right)-(1+\epsilon)\prod_{k=0}^{n-1} \frac{1+p^k}{p^k}
\geq 5.99-(1+\epsilon)\cdot 4.769
> 0.028,\qquad p=2,\ n\geq 10,
\\
2\left( 1 + \frac{p-p^{-n+1}}{p-1} \right)-(1+\epsilon)\prod_{k=0}^{n-1} \frac{1+p^k}{p^k}
\geq 4-(1+\epsilon)\cdot 3.13
> 0.08,\qquad p\geq 3.
\end{gather*}
\end{proof}

\section{Streaming}\label{sec:stream}
In all the linear algebra problems we have discussed, Alice and Bob want to know some property of the $n\times n$ matrix $x+y$, where Alice holds $x$ and Bob holds $y$. However, if we want to compute the matrix property in the streaming model, the matrix in the stream is represented in row order. Hence in the communication model, we need a different way to distribute the inputs: Alice holds the top half of the $n\times n$ matrix and Bob holds the bottom half. Formally, Alice holds an $(n/2)\times n$ matrix $x'$ and Bob holds another $(n/2)\times n$ matrix $y'$. They want to solve the problem on the concatenated matrix $\left(\begin{smallmatrix}
x' \\
y' \\
\end{smallmatrix}\right)$. In some circumstances, this is a more natural way to distribute the input. We shall show that all of our lower bounds still hold even in this setting.
\begin{theorem}
    Alice and Bob hold $(n/2)\times n$ matrices $x'$ and $y'$ over $\mathbb{F}_p$, respectively. They need $\Omega(n^2 \log p)$ qubits of communication even if they use a quantum protocol to compute \scrank$_{n,n-1}'$, $\scinv{'}$ or $\scsls'_b$.
\begin{itemize}
    \item \scrank$_{n,n-1}'$. Decide if $\rankp \left(\begin{smallmatrix}
x' \\
y' \\
\end{smallmatrix}\right)$ is $n$ or $n-1$;
    \item $\scinv'$. Decide if the top-left entry of $\left(\begin{smallmatrix}
x' \\
y' \\
\end{smallmatrix}\right)^{-1}$ is zero;
    \item \scsls$'_b$. $b\in \mathbb{F}_p^{n}$ is a non-zero vector. $t\in \mathbb{F}_p^n$ is the variable of the linear system $\left(\begin{smallmatrix}
x' \\
y' \\
\end{smallmatrix}\right)t=b$. They want to decide if $t_0$ is zero.
\end{itemize}
\end{theorem}
\begin{proof}
    We reduce the original problems to these problems.
    \begin{itemize}

    \item{\scrank$_{n,n-1}$.} Let $x'=\begin{pmatrix} x & -I \end{pmatrix}$, and $y'=\begin{pmatrix} y & I \end{pmatrix}$.
    $$
    \rankp(x+y) 
=   \rankp\begin{pmatrix}
    x+y & 0 \\
    y & I \\
    \end{pmatrix} - n
=   \rankp\begin{pmatrix}
    x & -I \\
    y & I \\
    \end{pmatrix} - n
=   \rankp\begin{pmatrix}
    x' \\
    y' \\
    \end{pmatrix} - n.
$$
    Therefore, $\rankp(x+y)=n$ iff $\rankp\left(\begin{smallmatrix}
x' \\
y' \\
\end{smallmatrix}\right) = 2n$, and $\rankp(x+y)=n-1$ iff $\rankp\left(\begin{smallmatrix}
x' \\
y' \\
\end{smallmatrix}\right) = 2n-1$.
    \item{$\scinv'$.} The reduction from \scrank$_{n,n-1}'$ almost works. The only problem is that the parity of the size of the matrix changes after appending one row and one column. To make the size even, we add another row on the bottom and another column on the right. The additional row and column are all zero except for the bottom-right entry, and the bottom-right entry is one. 
    \item{\scsls$'_b$.} The reduction from $\scinv'$ still works.\qedhere
    \end{itemize}
\end{proof}

Finally, this implies the streaming part of Theorem~\ref{thm:rank-lb}.

\section{Privacy}\label{sec:privacy}

We consider the $\textsc{rank}_{n,n-1}$ problem in this section. Let $\mu$ denote the uniform distribution over $G\times G$, where $G$ is the semi-group considered in Section~\ref{sec:multiparty}. Similarly to Theorem~\ref{thm:hardest_dist}, we have
\begin{theorem} Suppose that $\delta < 1/9$ when $p=2$ and $\delta < \frac{p}{3(1-p+p^2)}$ when $p\geq 3$. Then
\[
R_{1/3}^{\pub}(\textsc{Rank}_{n,n-1})\leq D_{\delta}^{\mu}(\textsc{rank}_{n,n-1})\leq C\left(\log \frac{1}{\delta}\right) \cdot R_{1/3}^{\pub}(\textsc{Rank}_{n,n-1}).
\]
\end{theorem}
\begin{proof}
One can verify that 
\[
\alpha := \frac{\Pr\{\rk(A+B)=n\}}{\Pr\{\rk(A+B)=n-1\}} = \left(1+\frac{1}{p^n-1}\right)\frac{(p-1)^2}{p}.
\] 
Note that $\alpha \approx 1/2 < 1$ when $p=2$ and $\alpha > 1$ when $p\geq 3$. Following the same reduction the same reduction in the proof of Theorem~\ref{thm:hardest_dist} with $\Pi_\nu$ replaced by $\Pi_\mu$, we conclude that the public coin protocol $\Pi'$ has error probability at most $1-\frac{(1+\alpha)(1-\delta)-1}{\alpha} < \frac 13$ when $p=2$ and at most $(1+\alpha)\delta < \frac 13$ when $p\geq 3$. The rest follows similarly as in Theorem~\ref{thm:hardest_dist}.
\end{proof}
As a corollary of \cite[Theorem 1.3]{BBCR10}, we know that when $p$ is a constant,
\[
\IC_\mu(\textsc{Singularity}_n)\geq \IC_\mu(\textsc{Rank}_{n,n-1}) =\Omega\left(\frac{D^\mu(\textsc{Rank}_{n,n-1})}{\polylog D^{\mu}(\textsc{Rank}_{n,n-1})}\right) = \Omega\left(\frac{n^2}{\polylog n}\right).
\]
We remark that combining \cite{JK10} and \cite{KLLRX12} yields $\max_\lambda \IC_\lambda(\textsc{Rank}_{n,n-1}) = \Omega(n^2)$, but it is not clear what distribution $\lambda$ attains the lower bound. Our bound above, although slightly weaker, shows that the product distribution nearly achieves the desired lower bound. Finally, it then follows from \cite[Proposition 20]{ACCFKP12} that
\[
\PRIV_\mu(\textsc{Singularity}_n) = \Omega\left(\frac{n^2}{\polylog n}\right).
\]

\end{document}